\documentclass[11pt]{article}
\usepackage{geometry}
\usepackage{amssymb}
\usepackage{mathrsfs}
\usepackage{amsmath}
\usepackage{float}
\usepackage{color}

\newcommand{\cout}[1]{}

\newtheorem{theorem}{Theorem}[section]

\newtheorem{lemma}[theorem]{Lemma}

\newtheorem{observation}[theorem]{Observation}

\def\squarebox#1{\hbox to #1{\hfill\vbox to #1{\vfill}}}
\newcommand{\qed}{\hspace*{\fill}
	\vbox{\hrule\hbox{\vrule\squarebox{.667em}\vrule}\hrule}\smallskip}

\newenvironment{proof}{\noindent{\bf Proof:~~}}{\(\qed\)}
\usepackage{graphicx}

\newcommand{\FS}{{\mathcal{F}^H_s}}
\newcommand{\FB}{{\mathcal{F}^H_b}}
\newcommand{\ERS}{{\mathcal{ER}^H_s}}
\newcommand{\ERB}{{\mathcal{ER}^H_b}}

\newcommand{\HH}{{\mathcal{H}}}

\newcommand{\comment}[1]{}
\begin{document}
	\title{A Note on the Gains from Trade of the Random-Offerer Mechanism}
		\author{Moshe Babaioff \thanks{Microsoft Research. Email: moshe@microsoft.com.} \and Shahar Dobzinski\thanks{Weizmann Institute of Science and Microsoft Research. Email: shahar.dobzinski@weizmann.ac.il. Partially supported by BSF grant 2016192, ISF grant 2185/19, and an NSF-BSF grant 2021655.} \and Ron Kupfer\thanks{Harvard University. Work done while an intern at Microsoft Research.  Email: ron.kupfer@mail.huji.ac.il.}	}
	\maketitle	
	\begin{abstract} 
We study the classic bilateral trade setting. Myerson and Satterthwaite show that there is no Bayesian incentive compatible and budget-balanced mechanism that obtains the gains from trade of the first-best mechanism. Consider the random-offerer mechanism: with probability $\frac 1 2$ run the \emph{seller-offering} mechanism, in which the seller offers the buyer a take-it-or-leave-it price that maximizes the expected profit of the seller, and with probability $\frac 1 2$ run the \emph{buyer-offering} mechanism. Very recently, Deng, Mao, Sivan, and Wang showed that the gains from trade of the random-offerer mechanism is at least a constant factor of $\frac 1 {8.23}\approx 0.121$ of the gains from trade of the first best mechanism. Perhaps a natural conjecture is that the gains-from-trade of the random-offerer mechanism, which is known to be at least half of the gains-from-trade of the second-best mechanism, is also at least half of the gains-from-trade of the first-best mechanism. However, in this note we exhibit distributions such as the gains-from trade of the random-offerer mechanism is smaller than a $0.495$-fraction of the gains-from-trade of the first-best mechanism.
\end{abstract}
	
\section{Introduction}

Consider a bilateral trade setting with a seller that holds one indivisible item and a buyer that is interested in buying that item. The value of the seller for keeping the item is $s\geq 0$ (and zero if he gives the item) and the value of the buyer if he obtains the item is $b\geq 0$ (and zero otherwise). The values $s$  and $b$ are drawn independently from distributions $\FS$ and $\FB$, respectively. The distributions $\FS$ and $\FB$ are known to both players, yet values are private.  

The optimal gains from trade (GFT) measures the maximum possible expected welfare improvement in the market due to trade. That is, the \emph{optimal GFT} (or the GFT of the \emph{first best} mechanism) is obtained by a trade happening if and only if $b\geq s$, it expected GFT is denoted by $OPT_{\FS, \FB}=E_{s\sim \FS, b\sim \FB}(b-s)_+=E_{s\sim \FS, b\sim \FB}[\max \{b-s,0\}]$. 

Myerson and Satterthwaite \cite{myerson1983efficient} consider distributions $\FS,\FB$ that are supported on $[\underline s,\overline s]$ and $[\underline b,\overline b]$, respectively, and have  continuous and positive density everywhere in their domains. 
They prove that if $\underline b<\overline s$ then any Bayesian incentive compatible and interim individually rational mechanism that is also ex-ante budget balanced (i.e., the expected payment that the buyer pays is at least the expected amount that the seller receives) never guarantees full efficiency (optimal GFT). They provide a formula for computing the mechanism that maximizes the GFT among all incentive compatible, individually rational and budget-balanced mechanisms (the \emph{second-best} mechanism). 

Providing a formula for the second-best mechanism is obviously a remarkable achievement. However, this formula includes computing integrals that are often hard to analytically compute. Thus, typically, for a given pair of distributions $\FS,\FB$ it is practically impossible to implement the second-best mechanism. Even when they can be computed, second-best mechanisms are often unintuitive and take a complex form. 


\subsection*{Some Related Work}
Soon after the result of \cite{myerson1983efficient}, economists have realized the deficiencies of second-best mechanisms and the search for alternatives has begun almost immediately. For example, Chatterjee and Samuelson \cite{chatterjee1983bargaining} offer the following \emph{k-double auction}: the seller submits a bid $r_s$, the buyer submits $r_b$. If $r_b<r_s$ the sell price is set to be $k\cdot r_b+(1-k)r_s$. For the case of $k\in(0,1)$, Satterthwaite and Williams \cite{SATTERTHWAITE1989107} and Leininger, Linhart and Radner \cite{leininger1989equilibria} show that the auction (which is not incentive compatible) may exhibit a continuum of equilibria. The efficiency of these equilibria vary significantly: some of these mechanisms might guarantee no efficiency at all, whereas others might guarantee the efficiency of second-best mechanisms.

In addition, Leininger, Linhard, and Radner \cite{leininger1989equilibria} consider several specific distributions and compute the ratio between the optimal GFT and the GFT obtained by the second-best mechanism for these distributions. They exhibit a specific pair of distributions in which this ratio is $\frac 2 e\approx 0.735$. 

McAfee \cite{mcafee2008gains} takes a different approach. Instead of struggling with the hard task of explicitly computing the second-best mechanism, McAfee suggests the Median mechanism: let $p$ be the median of the seller distribution $\FS$. Trade occurs if and only if $b\geq p \geq s$. McAfee shows that the median mechanism guarantees in expectation at least half of the optimal GFT, but only under the constraint that the median of $\FS$ is at most the median of $\FB$. Blumrosen and Dobzinski \cite{blumrosen2014reallocation} prove that in a sense this constraint is necessary, since there exists a pair of distributions for which no distribution-dependent price can guarantee a constant fraction of the optimal GFT.

Blumrosen and Mizrahi \cite{blumrosen2016approximating} study a different class of mechanisms.  Specifically, they introduce the \emph{seller-offering} mechanism in the context of bilateral trade. In this mechanism, a seller with value $s$ makes a take-it-or-leave-it offer $p$ for a price $p$ that maximizes his profit: $\Pr_{b\sim \FB}[b\geq p]\cdot (p-s)$. The buyer agrees only if $b\geq p$. Blumrosen and Mizrahi show that if $\FB$ is a distribution with a monotone hazard rate, the GFT is $1-\frac 1 e$ of the optimal GFT. 
Interestingly, the seller-offering mechanism can be seen as a $k$-double auction with $k=0$.

Brustle et al. \cite{brustle2017approximating} extend this mechanism and consider the \emph{random-offerer} mechanism: with probability half  it runs the seller-offering mechanism and with probability half it runs the 
analogous buyer-offering mechanism (a $k$-double auction with $k=1$): 
a buyer with value $b$ makes a take-it-or-leave-it offer $p$ for the price $p$ that maximizes his profit: $\Pr_{s\sim \FS}[s\leq p]\cdot (b-p)$. Analyzing the dual of a certain LP problem, they provide, among other results, an involved proof that shows that the GFT of this mechanism extracts is half of the GFT of the second-best mechanism.\footnote{Segal \cite{Segal} offered an alternative proof for this result that is short enough to fit in a footnote: the GFT of any BIC, IR and strongly budget-balanced mechanism is simply the sum of the profits of both players. Since the seller-offering mechanism provides the highest possible profit for the seller among all incentive compatible mechanism, and since similarly the buyer-offering mechanism provides the highest possible profit for the buyer, the sum of profits of both mechanisms is at least the GFT of any BIC, IR and strongly-budget mechanism, and the result follows.}

\subsection*{Our Result}

Only very recently, Deng, Mao, Sivan, and Wang \cite{deng2021approximately} were able to show that the GFT of the second-best mechanism is a constant factor of the gains from trade of the first-best mechanism. Specifically, they show that the random-offerer mechanism always guarantees a constant factor of $\frac 1 {8.23}$ of the gains from trade of the first-best mechanism.

In this note we further study the performance of the random-offerer mechanism. Given the result of Brustle et al. \cite{brustle2017approximating}, one could perhaps conjecture that the gains-from-trade of the random-offerer mechanism is always at least half of the gains from trade of the first-best mechanism, not just of that of the second-best mechanism. However, we show that:

\begin{theorem}\label{thm-main}
There exists a distribution of the seller $\FS$ and a distribution of the buyer $\FB$ such that the gains from trade of the random-offerer mechanism is less than $0.495$-fraction of the gains from trade of the first-best mechanism.
\end{theorem}
We note that both distributions are finite and have discrete support (extending our result to the continuous case is straightforward).




\subsection*{Future Directions}

One basic open question is to determine how large can be the ratio between the GFT of the first-best mechanism and the GFT of the second-price mechanism. There are distributions for which this ratio is at most $\frac 2 e$ \cite{blumrosen2016approximating,leininger1989equilibria}. Determining the exact fraction of the optimal GFT that second-best mechanisms can always guarantee is an important open question.

However, the second-best mechanism is often quite complex, whereas the random-offerer mechanism is natural and simple. Thus, it is very interesting to quantify the additional power that the complexity of the second-best mechanism gives: how large can the ratio between the GFT of the first-best mechanism and the GFT of the random-offerer mechanism be? 

It is also interesting to understand the power of seller/buyer-offering mechanisms with respect to the second-best mechanism. We know of some distributions for which the GFT of the random-offerer mechanism is no more than half of the GFT of the second-best mechanism.\footnote{Consider the distribution in which the value of the seller is identically $0$, whereas the buyer's value is distributed by a distribution $\mathcal D$ that is equal to the equal revenue distribution $\mathcal D'$ in the interval $[1,H)$ and $\Pr_{b\sim \mathcal D[b=H}=\Pr_{b\sim \mathcal D'[b \geq H]}$. The buyer-offering mechanism GFT equal the first best whereas the GFT of the seller-offering mechanism (which always offers $H$) approaches $0$ as $H$ goes to infinity. Note that the GFT of the buyer-offering mechanism -- which always makes an offer of $0$ -- is optimal for these distributions.} However, for these distributions, the ``best of'' mechanism, which runs the mechanism with the higher GFT of the two offering mechanisms, is the second-best mechanism (note that in the distributions used in the proof of Theorem \ref{thm-main} the GFTs of the seller-offering mechanism and the buyer-offering mechanism are identical, therefore the performance of the random-offerer mechanism and the ``best of'' mechanism is the same). In Section \ref{sec-best-of} we observe 
that there are distributions for which the ``best of'' mechanism provides no more than $\frac 3 4$ of the GFT of the second-best mechanism. Thus, there is a gap between this bound of $\frac 3 4$ and the guarantee of $\frac 1 2$ ensured by the random-offerer (and the ``best of'') mechanism.
	\section{The Hard Distributions (Proof of Theorem \ref{thm-main})}\label{sec-main}

We consider distributions whose support contains only in the set of integers $\{0, 1,\ldots, H\}$ for $H>3$. The seller's value distribution $\FS$ and the buyer's distribution are:
$$
\Pr_{s\sim \FS}(s\leq m)=
\begin{cases}
\frac 1 {2(H-1)} & m=0, \\
\frac 1{ H-m} & m \in \{1,2,\ldots,H-1\}, \\
1 & m=H.
\end{cases}
\quad \quad
\Pr_{b\sim \FB}(b\geq m)=
\begin{cases}
\frac 1 {2(H-1)} & m=H, \\
\frac 1 m & m \in \{1,2,\ldots,H-1\}, \\
1 & m=0.
\end{cases}
$$


In the proof we assume that when an agent makes an offer, we can determine the way ties are broken. A specific tie-breaking rule can be iplemented by shifting an arbitrarily small probability mass from one tied price to another. The effect of such a shift on the GFT is negligible. 

We make the following simple observations. 
A seller with a value $s=0$ is indifferent offering any integer price in $\{1, 2, \ldots, H-1\}$. For $s\in [1,H-3]$ the seller offers a price $H-1$, which is the unique profit maximizing price.
For $s=H-2$, the seller is indifferent between offering $H-1$ or $H$, so similarly we assume it offers $H$. For $s=H-1$  the seller's offer is $H$.

We note that the distributions of the buyer and the seller are symmetric (the roles of 0 and $H$ are switched). Thus, we assume that a buyer with a value $b=H$ offers $1$.
For $b\in \{2,\ldots, H-1\}$ the buyer offers a price $1$.
For $b=2$, the buyer is indifferent between offering $1$ or $0$ and we assume it offers $0$.
For $b=1$ the seller offers a price $0$.


We denote the GFT of the buyer-offering mechanism by $BO_{\FS, \FB}$, and  the GFT of the seller-offering mechanism by $SO_{\FS, \FB}$. The GFT of the random-offerer mechanism is $(BO_{\FS, \FB}+SO_{\FS, \FB})/2$.
To prove the claim we compute $OPT_{\FS, \FB}$ as well as $BO_{\FS, \FB}$ and $SO_{\FS, \FB}$. 
To do so, we first compute the optimal GFT for two very close distributions.

\subsection{Discrete Equal Revenue Distributions}\label{sec:ER}

The distributions $\FB,\FS$ are very similar to two distributions that are easier to analyze. We first analyze the optimal GFT and the GFT of the buyer-offering and seller-offering mechanisms for there distributions, and then compute the GFT of these mechanisms on the distributions $\FB$ and $\FS$ by examining the difference between the distributions. The two Discrete Equal Revenue Distributions $\ERS$ and $\ERB$ are defined as follows: 
$$
\Pr_{s\sim \ERS}(s\leq m)=
\begin{cases}
\frac 1{ H-m} & m \in \{0, 1,\ldots,H-1\}, \\
1 & m=H.
\end{cases}
\quad \quad
\Pr_{b\sim \ERB}(b\geq m)=
\begin{cases}
1 & m=0,\\
\frac 1 m & m \in  \{1, 2, \ldots,H\}. \\
\end{cases}
$$
The buyer distribution $\ERB$ has support $X=\{1,2,\ldots,H\}$ for $H\geq 2$. For $x\in X$, the probability that $b\geq x$ is $1/x$. 
Note that the probability that $b=H$ is $1/H$, and for any $b\in X\setminus \{H\}$, the probability that $b=x$ is $\frac{1}{x(x+1)}$.

The seller distribution $\ERS$ has support $Y=\{0,1,\ldots,H-1\}$. For $y\in Y$, the probability that $s\leq y$ is $1/(H-y)$. Note that the probability that $s=0$ is $1/H$, and for any $s\in Y\setminus \{0\}$, the probability that $s=y$ is $\frac{1}{(H-y)(H-y+1)}$. In the rest of this section we prove the following claim:
\begin{lemma}\label{lem:opt-ERs}
    It holds that $OPT_{\ERS, \ERB}=\frac{ 2\cdot \HH_{H}}{H+1} $. 
\end{lemma}
\begin{proof}
Denote $OPT_{\ERS, \ERB}$ by OPT. Observe that:
\begin{align*}
OPT &= H\cdot \Pr[s=0]\cdot \Pr[b=H] + E_{b\sim \ERB, b<H}[b-0]\cdot \Pr[s=0] + E_{s\sim \ERS, s>0}[H-s]\cdot \Pr[b=H] \\
& + E_{s\sim \ERS,b\sim \ERB, s> 0, b< H} [(b-s)_+] \\
&= \frac{1}{H} + \frac{2 ({\HH}_{H}-1)}{H}  + E_{s\sim \ERS,b\sim \ERB, s> 0, b<H} [(b-s)_+]\\ 
&=  \frac{2\cdot {\HH}_{H}-1}{H}  + E_{s\sim \ERS,b\sim \ERB, s> 0, b< H} [(b-s)_+]
\end{align*}
\begin{align*}
E_{s\sim \ERS,b\sim \ERB, s> 0, b< H} [(b-s)_+] &=
\sum_{s=1}^{H-1} \sum_{b=s+1}^{H-1} (b-s)\cdot \frac{1}{b(b+1)}\cdot \frac{1}{(H-s)(H-s+1)} \\
&= \sum_{s=1}^{H-1}  \frac{1}{(H-s)(H-s+1)} \sum_{b=s+1}^{H-1} (b-s)\cdot \frac{1}{b(b+1)} 
\end{align*}
Observe that 
\begin{align*}
\sum_{b=s+1}^{H-1} (b-s)\cdot \frac{1}{b(b+1)} &=
\sum_{b=s+1}^{H-1} \frac{1}{b+1} - s \sum_{b=s+1}^{H-1} \left(\frac{1}{b}- \frac{1}{b+1}\right)\\
&=\HH_H - \HH_{s+1} - s(\HH_{H-1} - \HH_{s} -(\HH_H - \HH_{s+1}))\\
&= \HH_H - \HH_{s+1} -\frac{s}{s+1} +\frac{s}{H} \\ 
&= \HH_H - \HH_{s+1} - 1 +\frac{1}{s+1} +\frac{s}{H} \\
&= \HH_H - \HH_{s} -\frac{H-s}{H} \\
\end{align*}
thus
\begin{align*}
E_{s\sim \ERS,b\sim \ERB, s> 0, b< H} [(b-s)_+] &=
\sum_{s=1}^{H-1}  \frac{\HH_H - \HH_{s} -\frac{H-s}{H}}{(H-s)(H-s+1)} \\
&= \sum_{s=1}^{H-1}  \frac{\HH_H - \HH_{s}}{(H-s)(H-s+1)} -
\sum_{s=1}^{H-1}  \frac{1}{H(H-s+1)} \\
&= \sum_{s=1}^{H-1}  (\HH_H - \HH_{s}) \left(\frac{1}{H-s}- \frac{1}{H-s+1}\right) -
\frac{1}{H} \cdot (\HH_H-1) \\
&= \HH_H\cdot \sum_{s=1}^{H-1} \left(\frac{1}{H-s}- \frac{1}{H-s+1}\right) 
-\sum_{s=1}^{H-1}  \HH_{s} \left(\frac{1}{H-s}-\frac{1}{H-s+1} \right)
- \frac{\HH_H-1}{H} \\
&= \HH_H\left(1-\frac{1}{H}\right) 
-
\sum_{s=1}^{H-1}  \HH_{s} \left(\frac{1}{H-s} -\frac{1}{H-s+1} \right)
- \frac{\HH_H-1}{H} \\
\end{align*}
\begin{align*}
\sum_{s=1}^{H-1}  \HH_{s} \left(\frac{1}{H-s}- \frac{1}{H-s+1} \right) &=
\sum_{s=1}^{H-1}   \frac{\HH_{s}}{H-s}- 
\sum_{s=1}^{H-1}   \frac{\HH_{s}}{H-s+1}  \\
&= \HH_{H-1}+ \sum_{s=1}^{H-2}   \frac{\HH_{s}}{H-s}- 
\sum_{s=2}^{H-1}   \frac{\HH_{s}}{H-s+1} -\frac{1}{H} \\
&= \HH_{H-1}-\frac{1}{H} + \sum_{s=1}^{H-2}   \frac{\HH_{s}}{H-s}-  \sum_{s=1}^{H-2}   \frac{\HH_{s+1}}{H-s}  \\
&= \HH_{H-1}-\frac{1}{H}-\sum_{s=1}^{H-2}   \frac{1}{s+1}\cdot \frac{1}{H-s}  \\
&=\HH_{H-1}-\frac{1}{H} - \frac{1}{H+1}\sum_{s=1}^{H-2}  \left( \frac{1}{s+1}+ \frac{1}{H-s} \right) \\
&=\HH_{H-1}-\frac{1}{H} - \frac{2(\HH_{H-1} -1)}{H+1} 
\end{align*}
Combining the two we get:
\begin{align*}
E_{s\sim \ERS,b\sim \ERB, s> 0, b< H} [(b-s)_+]
&= \HH_H\left(1-\frac{1}{H}\right) 
-
\sum_{s=1}^{H-1}  \HH_{s} \left(\frac{1}{H-s} -\frac{1}{H-s+1} \right)
- \frac{\HH_H-1}{H} \\
&= \HH_H-\frac{\HH_H}{H} 
-
\left(\HH_{H-1}-\frac{1}{H} - \frac{2(\HH_{H-1} -1)}{H+1}  \right)
- \frac{\HH_H}{H} +\frac{1}{H} \\
&= \HH_H-\frac{2\HH_H}{H} 
-\HH_{H-1}+\frac{1}{H} + \frac{2(\HH_{H-1} -1)}{H+1} 
+\frac{1}{H} \\
&= \frac{3}{H} -\frac{2\HH_H}{H}  + \frac{2(\HH_{H-1} -1)}{H+1} 
\\
\end{align*}
We thus conclude that 
\begin{align*}
OPT  
&=  \frac{2\cdot {\HH}_{H}-1}{H}  + E_{s\sim \ERS,b\sim \ERB, s> 0, b< H} [(b-s)_+]\\
&=  \frac{2\cdot {\HH}_{H}-1}{H}  + \frac{3}{H} -\frac{2\HH_H}{H}  + \frac{2(\HH_{H-1} -1)}{H+1}  \\
&=  \frac{2}{H}  + \frac{2(\HH_{H-1} -1)}{H+1}  \\
&=  \frac{2H+2 + 2 H \HH_{H-1} -2 H}{H(H+1)}  \\
&=  \frac{ 2 (H \HH_{H-1}+1)}{H(H+1)}  \\
&=  \frac{ 2 H \cdot\HH_{H}}{H(H+1)}  \\
&=  \frac{ 2\cdot \HH_{H}}{H+1}  \\
\end{align*}
\end{proof}

\subsection{Back to $\FS$ and $\FB$}

To prove the theorem we use the above bound to prove two claims:

\begin{lemma}\label{lem:opt}
    It holds that $OPT_{\FS, \FB}= OPT_{\ERS, \ERB}-\frac{H-2}{H(H-1)}=\frac{ 2\cdot \HH_{H}}{H+1} -\frac{H-2}{H(H-1)}$.
\end{lemma}
\begin{proof}
    By Lemma \ref{lem:opt-ERs} it holds that $OPT_{\ERS, \ERB}=\frac{ 2\cdot \HH_{H}}{H+1} $. Observe that $\FB$ is the distribution  obtained by starting with  $\ERB$ and moving probability  mass of size 
    $\frac{1}{H}-\frac{1}{2(H-1)}$ from $H$ to $H-1$. Similarly, $\FS$ is the distribution  obtained by starting with  $\ERS$ and moving probability  mass of size
    $\frac{1}{H}-\frac{1}{2(H-1)}$ from $0$ to $1$.
    Note that in $OPT_{\ERS, \ERB}$, the GFT includes the trade of every seller with a buyer of value $H$, and moving probability mass of  $\frac{1}{H}-\frac{1}{2(H-1)}$ from $H$ to $H-1$ decreases the GFT by exactly  $\frac{1}{H}-\frac{1}{2(H-1)}$. Similarly, the GFT decreases by the same amount due to the move in the seller mass. The claim follow as the total GFT loss is $2(\frac{1}{H}-\frac{1}{2(H-1)})= \frac{H-2}{H(H-1)}$.
\end{proof}
\begin{lemma}\label{lem:bo}
   For $H>2$ it holds that {$BO_{\FS, \FB}=SO_{\FS, \FB}=\frac{\HH_{H-1}}{H-1} -     \frac{7}{12(H-1)} + \frac{1}{2(H-1)^2}$}.
\end{lemma}
\begin{proof}
    Clearly, by symmetry, $BO_{\FS, \FB}=SO_{\FS, \FB}$. We thus compute $BO_{\FS, \FB}$. 
    Any buyer with $b>2$ offers price of $1$, buyers with value $b\in \{1,2\}$ offer $0$. 
    The GFT of the buyer-offering mechanism is thus:
\begin{align*}
    BO_{\FS, \FB} = & 
    E_{b>2,b\sim \FB} [b-1]\cdot \Pr[s=1] + 
    E_{b>2,b\sim \FB} [b]\cdot \Pr[s=0] + 
    \Pr[b=1,s=0]+2\Pr[b=2,s=0]
     \\ = & 
    E_{b>2,b\sim \FB} [2b-1]\cdot \frac{1}{2(H-1)} + 
    \left(\frac{1}{2}+ 2\cdot \frac{1}{6}\right) \frac{1}{2(H-1)}
    \\ = & 
    \frac{1}{2(H-1)} \left( E_{b>2,b\sim \FB} [2b-1] + \frac{5}{6} \right)
\end{align*}

We next compute $E_{b>2,b\sim \FB} [2b-1]$.
\begin{align*}
 E_{b>2,b\sim \FB} [2b-1] = &
 \sum_{b=3}^{H-2} (2b-1)\cdot \frac{1}{b(b+1)} + (2(H-1)-1)\cdot \frac{1}{2(H-1)} + (2H-1)\cdot \frac{1}{2(H-1)} \\ = &
  \sum_{b=3}^{H-2} 2b\cdot \frac{1}{b(b+1)} -  \sum_{b=3}^{H-2} \left(\frac{1}{b}- \frac{1}{b+1}\right)+ 2\\ = &
  2\sum_{b=3}^{H-2} \frac{1}{b+1} -  \left(\frac{1}{3}- \frac{1}{H-1}\right)+ 2
  \\ = &
  2\left(\HH_{H-1} -1 - \frac{1}{2} - \frac{1}{3}\right) -  \frac{1}{3}+ \frac{1}{H-1}+ 2
  \\ = &
  2\HH_{H-1} +\frac{1}{H-1}-2
\end{align*}
We conclude that
\begin{align*}
    BO_{\FS, \FB} = &  
    \frac{1}{2(H-1)} \left( E_{b>2,b\sim \FB} [2b-1] + \frac{5}{6} \right)
    \\ = &
    \frac{1}{2(H-1)} \left( 2\HH_{H-1} +\frac{1}{H-1}-2 + \frac{5}{6} \right)
       \\ = &
    \frac{\HH_{H-1}}{H-1} -     \frac{7}{12(H-1)} + \frac{1}{2(H-1)^2} 
\end{align*}

\end{proof}

We are now ready to complete the proof of Theorem \ref{thm-main}. The theorem will directly follow from the following corollary derived from the two lemmas above.

\begin{lemma}\label{lem:below-half}
For the distribution $\FS,\FB$ with $H=1000$, the ratio between the optimal GFT and the GFT of the random-offerer mechanism is smaller than $0.495$. That is: 
$$\frac{BO_{\FS, \FB}+SO_{\FS, \FB}}{2\cdot OPT_{\FS, \FB}}< {0.495}$$
\end{lemma}
\begin{proof}
To prove the claim we calculate the ratio using the two expressions proven in 
Lemma \ref{lem:opt} and Lemma \ref{lem:bo}. 
By Lemma \ref{lem:opt} it holds that:
    $$OPT_{\FS, \FB}=\frac{ 2\cdot \HH_{H}}{H+1} -\frac{H-2}{H(H-1)}$$
By Lemma \ref {lem:bo} for $H>2$ 
    it holds that:
    $$BO_{\FS, \FB}=SO_{\FS, \FB}=\frac{\HH_{H-1}}{H-1} -     \frac{7}{12(H-1)} + \frac{1}{2(H-1)^2}$$  
Plugging in $H=1000$ we get that:
$$\frac{BO_{\FS, \FB}+SO_{\FS, \FB}}{2\cdot OPT_{\FS, \FB}}= \frac{\frac{\HH_{999}}{999} -     \frac{7}{12\cdot 999} + \frac{1}{2\cdot{999}^2}}{\frac{ 2\cdot \HH_{1000}}{1001} -\frac{998}{1000\cdot 999}} < 0.495$$
\end{proof}

In Figure \ref{fig:ratio} we plot the ratio for different values of $H$. The ratio is larger than $\frac 1 2$ for small values of $H$, then goes down to below $\frac 1 2$ (Lemma \ref{lem:below-half}). As $H$ goes to infinity the ratio
approaches $\frac 1 2$ from below, see Appendix \ref{app:obs-half}. 

\begin{figure}
    \centering
    \includegraphics{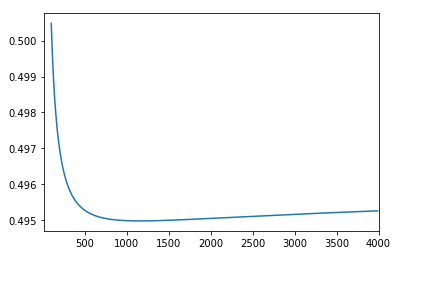}
    \caption{$\frac{BO_{\FS, \FB}+SO_{\FS, \FB}}{2\cdot OPT_{\FS, \FB}}$ as a function of $H$, for $H\in[100,4000]$.}
    \label{fig:ratio}
\end{figure}


\section{The ''Best-of'' Mechanism and the Second-Best Mechanism}\label{sec-best-of}

We consider the distributions $\ERB,\ERS$ of Section \ref{sec-main} with $H=2$. Observe that in a fixed-price mechanism that always offers a price of $1$, trade occurs whenever trade is profitable (when ties are broken in favour of trade), thus this mechanism extracts the optimal GFT. Therefore, for these distributions the GFT of the second-best mechanism and the GFT of the first-best mechanism are the same and equal $\frac {2\cdot (1+\frac 1 2)} 3-0=1$, by Lemma \ref{lem:opt-ERs}.

However, for these distributions, the GFT of each of  the seller-offering mechanism and the buyer-offering mechanism is $\frac 1 2 + (\frac 1 2)^2 = \frac{3}{4}$. We thus get that there are distributions for which  the GFT of the ``best-of'' mechanism  (the better of seller-offering and buyer-offering)  is at most $\frac 3 4$ of the GFT of the second-best mechanism. 
	\bibliography{main} 
	\bibliographystyle{plain}		
\appendix
\section{Convergence to $\frac 1 2$}\label{app:obs-half}

\begin{observation}\label{obs:half}
$\lim_{H\rightarrow\infty} \frac  {2\cdot (BO_{\FS, \FB}+SO_{\FS, \FB})}{OPT_{\FS, \FB}} = \frac{1}{2}$. In addition, for $H>120$, we have that $$\frac  {2\cdot (BO_{\FS, \FB}+SO_{\FS, \FB})}{OPT_{\FS, \FB}} < \frac{1}{2}$$
\end{observation}
\begin{proof}
\begin{align*}
\frac  {2\cdot (BO_{\FS, \FB}+SO_{\FS, \FB})}{OPT_{\FS, \FB}} =&
\frac{\frac{\HH_{H-1}}{H-1} -     \frac{7}{12(H-1)} + \frac{1}{2(H-1)^2}}{\frac{ 2\cdot \HH_{H}}{H+1} -\frac{H-2}{H(H-1)}}\\
=&\frac{\frac{1}{H-1}\left(\HH_{H}-\frac{1}{H}-\frac{7}{12}+\frac{1}{2(H-1)}\right)}{\frac{1}{H-1}\left(\frac{H-1}{H+1}\cdot2\HH_H-1+\frac{2}{H}\right)}\\
=&\frac{\HH_{H}-\frac{7}{12}-\frac{H-2}{2H(H-1)}}{\frac{H-1}{H+1}\cdot2\HH_H-1+\frac{2}{H}}\\
\leq&\frac{\HH_{H}-\frac{7}{12}}{\frac{H-1}{H+1}\cdot2\HH_H-1+\frac{2}{H+1}}\\
=&
\frac{\HH_H-\frac{7}{12}}{\left(\frac{H-1}{H+1}\right)(2\HH_{H}-1)}\\
=&
\left(\frac{H+1}{H-1}\right)\cdot\left(\frac{1}{2}-\frac{1}{12(2\HH_{H}-1)}\right)\\
=&
\left(1+\frac{2}{H-1}\right)\cdot\left(\frac{1}{2}-\frac{1}{12(2\HH_{H}-1)}\right)\\
=&
\frac{1}{2}+\frac{1}{H-1}-\frac{1}{12(2\HH_{H}-1)}-\frac{1}{12(H-1)(2\HH_{H}-1)}\\<&\frac{1}{2}+\frac{1}{H-1}-\frac{1}{12(2\HH_{H}-1)}<\frac{1}{2}
\end{align*}
where the first inequality holds for $H>2$ and the last holds for $H>120$.
In the other direction, we have that:
\begin{align*}
\frac  {2\cdot (BO_{\FS, \FB}+SO_{\FS, \FB})}{OPT_{\FS, \FB}} =&
\frac{\HH_{H}-\frac{7}{12}-\frac{H-2}{2H(H-1)}}{\frac{H-1}{H+1}\cdot2\HH_H-1+\frac{2}{H}}\\
\geq&
\frac{\HH_{H}-\frac{7}{12}-\frac{H-2}{2H(H-1)}}{2\HH_H+\frac{2}{H}}\\ 
=&\frac{1}{2}-\frac{\frac{1}{H}+\frac{7}{12}+\frac{H-2}{2H(H-1)}}{2\HH_H+\frac{2}{H}}
\end{align*}
Which also converges to $\frac 1 2$.
\end{proof}

\end{document}